\documentclass[journal,draftcls,onecolumn,12pt,twoside]{IEEEtranTCOM}
\usepackage{setspace}
\usepackage{multirow,booktabs}
\usepackage{graphicx}
\usepackage[cmex10]{amsmath}
\usepackage{amsmath}
\usepackage{amssymb}
\usepackage{bm}
\usepackage{amsthm}
\usepackage{color}

\newtheoremstyle{theroremthm}
{1pt}
{1pt}
{}
{}
{\bf}
{:}
{.5em}
{}

\theoremstyle{theroremthm}
\newtheorem{theorem}{Theorem}

\newtheorem{lemma}{Lemma}
\usepackage{amsfonts}
\usepackage{fmtcount}
\DeclareGraphicsExtensions{ps,eps,pst,PDF}
\usepackage{cite}

\usepackage{stfloats}
\IEEEoverridecommandlockouts
\begin{document}
%
\title{Unified Framework for the Effective Rate Analysis of Wireless Communication Systems over MISO Fading Channels}
%
\author{\IEEEauthorblockN{Minglei~You,~\IEEEmembership{Student Member,~IEEE,} Hongjian~Sun,~\IEEEmembership{Senior Member,~IEEE,}\\ Jing~Jiang,~\IEEEmembership{Member,~IEEE}~and~Jiayi~Zhang,~\IEEEmembership{Member,~IEEE}}
\thanks{The research leading to these results has received fundings from the European Commission's Horizon
2020 Framework Programme (H2020/2014-2020) under grant agreement No. 646470  SmarterEMC2 Project, the National Natural Science Foundation of China under grant agreement No. 61601020, the National High-tech Research and Development Program of China (863 Program) under grant agreement No. 2015AA01A705 and the Fundamental Research Funds for the Central Universities under grant agreement No. 2016JBZ003 and 2016RC013.}
\thanks{Minglei You, Hongjian Sun and Jing Jiang are with School of Engineering and Computing Sciences, Durham University, Durham DH13LE, UK, (e-mails:\{minglei.you, hongjian.sun, jing.jiang\}@durham.ac.uk).

Jiayi Zhang is with School of Electronics and Information Engineering,
 Beijing Jiaotong University, Beijing 100044, P. R. China, (e-mail: jiayizhang@bjtu.edu.cn). 
}
}
\markboth{IEEE Transactions on Communications}%
{Submitted paper}
\maketitle
\begin{abstract}
This paper proposes a unified framework for the effective rate analysis over arbitrary correlated and not necessarily identical multiple inputs single output (MISO) fading channels, which uses moment generating function (MGF) based approach and \textit{H} transform representation. The proposed framework has the potential to simplify the cumbersome analysis procedure compared to the probability density function (PDF) based approach. Moreover, the effective rates over two specific fading scenarios are investigated, namely independent but not necessarily identical distributed (i.n.i.d.) MISO hyper Fox's \textit{H} fading channels and arbitrary correlated generalized $K$ fading channels. The exact analytical representations for these two scenarios are also presented. By substituting corresponding parameters, the effective rates in various practical fading scenarios, such as Rayleigh, Nakagami-$m$, Weibull/Gamma and generalized $K$ fading channels, are readily available. In addition, asymptotic approximations are provided for the proposed \textit{H} transform and MGF based approach as well as for the effective rate over i.n.i.d. MISO hyper Fox's \textit{H} fading channels. Simulations under various fading scenarios are also presented, which support the validity of the proposed method.
\end{abstract}

\IEEEpeerreviewmaketitle
\section{Introduction}
When evaluating the maximum achievable bit rate of a wireless system over fading channels, Shannon's channel capacity is the most important performance metric. It has also been widely adopted as the basis of performance analysis as well as mechanism design. However, many emerging applications are real-time applications, for example voice over IP, interactive video and most of the smart grid applications. For these real-time applications, not only throughput, but also delay should be considered as one of the quality of service (QoS) requirements. Unfortunately, delay performance cannot be analysed using the traditional Shannon's theory. 

In order to deal with this issue, the theory of effective rate (or effective capacity) has been proposed by Wu and Negi \cite{EffCap}, which bridges the gap between statistical QoS guarantees and the maximum achievable transmission rate. Since then it has been widely used as a powerful analytical tool and QoS provisioning metric in different scenarios. 
In \cite{E2E_EffCap2012}, the resource allocation and flow selection algorithms for video distribution over wireless networks have been studied based on the effective rate theory, where energy efficiency and statistical delay bound have been considered. 
In \cite{EC_TD_Scheduling2010} and \cite{ref_ER_TD2015}, scheduling algorithms have been studied for the multi-user time division downlink systems, exploiting the effective rate as the key to characterizing QoS constraints.
In \cite{EC_2Hop2013}, the effective rate of two-hop wireless communication systems has been studied, where the impact of nodes' buffer constraints on the throughput has been considered.
The effective rate has also been applied in the research of cognitive radio networks \cite{EffCap_CR_2010, ref_EC_CRN2015} for assisting QoS analysis. 

The channel's fading effect is one major reason for the fluctuation of the instantaneous channel capacity, but how to analyse the effective rate under different fading scenarios is still an open research area. There exist a few successful attempts, such as the effective rate over Nakagami-$m$, Rician and generalized $K$ \cite{MISOFramework2012}, Weibull \cite{ER_Weibull2015}, 
$\eta{-}\mu$ \cite{eta_mu_EffRate_2013}
, $\alpha{-}\mu$ \cite{alpha_mu_EffRate_2015}, $\kappa{-}\mu$ shadowed channels \cite{kappa_mu_EffRate_2015}, 
correlated exponential channels\cite{correlated_MISO_EffRate}. 
In \cite{ref_correlated_Nakagami2012}, the effective rate over correlated Nakagami-$m$ channels has been obtained in closed form, where the effective rate over correlated Rician fading channel has also been derived in analytical form.

These approaches can be categorized as probability density function (PDF) based method, since the analyses highly rely on the exact or approximated PDF of the signal-to-noise ratio (SNR). A general PDF based framework for the multiple inputs single output (MISO) effective rate analysis has been proposed in \cite{MISOFramework2012}. However, the joint PDF is unavailable for many fading channels and it is often very hard, if not impossible, to obtain the exact PDF for further analysis. Hence many researches have put a lot of effort on approximating the PDF of multiple fading channels\cite{MISOFramework2012, alpha_mu_EffRate_2015, kappa_mu_EffRate_2015}. This fact results in the situation that the PDF based effective rate analyses are generally studied in a case-by-case way.

To address aforementioned issues, this paper proposes a moment generating function (MGF) based framework for the effective rate analysis over MISO fading channels using \textit{H} transform representation. 
It has been clearly pointed out in \cite{MGF_EGC_MRC_GFading_ShannonC_2012, ShannonC_MGF_2010,ref_MGF_Pout2000} that MGF based approaches are beneficial in simplifying the analysis or even enabling the calculations of some important performance indexes when the PDF based approaches seem impractical. 
There are successful trials to use the MGF based approach to analyse the effective rate under single input single output (SISO) conditions\cite{ER_SISO_MGF_H_2014} and multi-hop systems\cite{ref_MGF_ER_MultiHop_2016}. In addition, \textit{H} transform analysis method \cite{H_Transform2015} has been exploited in this paper. Besides effective rate, many important metrics in the wireless communication system, for example ergodic capacity, error probability and error exponents, are defined in integrations form. It can be observed from vast amount of literatures that it is non-trivial and usually complicated to deal with these integrations  \cite{MISOFramework2012,ER_Weibull2015,eta_mu_EffRate_2013,%
alpha_mu_EffRate_2015,kappa_mu_EffRate_2015,correlated_MISO_EffRate,
ref_correlated_Nakagami2012,
MGF_EGC_MRC_GFading_ShannonC_2012}, especially when multivariate random variable is involved. It has been pointed out in \cite{H_Transform2015} that \textit{H} transform is a potential unified analysis method, which can provide a systematic language in dealing with the statistical metrics involving random variables in the wireless communication systems. This aspect will be shown in later part of this paper.

The scope of this paper is effective rate over arbitrary correlated and not necessarily identical MISO fading channels. Note that when dealing with multiple channels, the instantaneous channel power gain at the receiver is described by multivariate random variable in general cases, which is different from the SISO conditions where univariate random variable may be sufficient. In \cite{ER_SISO_MGF_H_2014}, the SISO case is well studied, yet the results are hard to extend to multiple channel conditions, which is the main focus of this paper. Moreover, in this paper we use the \textit{H} transform and multivariate Fox's \textit{H} function to present the effective rate over both i.n.i.d. and correlated fading channels, which further simplify the analysis and provide a more general analytical framework.
The major contributions of this paper are listed as follows:

\begin{itemize}
\item A MGF based framework is proposed for analysing the effective rate over arbitrary correlated and not necessarily identical MISO fading channels using \textit{H} transform representation. Due to the properties of \textit{H} transform, the cumbersome mathematical calculation involving integration operation can be simplified.

\item \textit{H} transform involving multivariate Fox's \textit{H} function is investigated. As many important metrics in wireless communication systems can be represented by \textit{H} transform and the statistical properties of multiple fading channels may be characterized using multivariate Fox's \textit{H} functions, the obtained results are also valuable in the analysis of other metrics in the multiple channel conditions. Also to the authors' best knowledge, this work is the first to deal with multiple channel problems within the \textit{H} transform framework.
%
%

\item Effective rate over both i.n.i.d. and correlated channel scenarios are studied. The exact analytical expressions of effective rate over i.n.i.d. MISO hyper Fox's \textit{H} fading channels and 
arbitrary correlated generalized $K$ fading channels are given. The effective rates over various practical fading channels are readily available by simply substituting corresponding parameters, such as generalized $K$ and Weibull/Gamma fading channels, which avoids the case-by-case study in these fading scenarios.

\item Asymptotic approximations are provided for the effective rate analysis over MISO fading channels, 
where the truncation error and the discretization error are studied. Using this approximation, the proposed effective rate expressions over i.n.i.d. MISO hyper Fox's \textit{H} fading channels can be accurately estimated in a unified and closed-form framework.  
\end{itemize} 

The rest of this paper is organized as follows. System model is introduced in Section \ref{section System Model and Background Knowledge}. Section \ref{section effective rate general} proposes the MGF based approach for MISO effective rate over arbitrary correlated and not necessarily identical fading channels using \textit{H} transform representation. The effective rate over i.n.i.d. MISO hyper Fox's \textit{H} fading channel as well as arbitrary correlated generalized $K$ fading channels is investigated in Section \ref{section ER H and hyper H}, where several special cases are discussed. Then approximations are investigated in Section \ref{section approximation}, 
with simulation results presented in Section \ref{section simulation}. Finally, conclusions are drawn in Section \ref{section conclusion}.

Throughout the paper, the following notations are used.
$\mathbb{C}$ denotes complex numbers. Bold letters denote vectors.
$\mathbf{1}_{N}$ denotes all-one vector of $N$ elements.
 $(\cdot)^{\dagger}$ denotes the Hermitian transpose.
$\mathbb{E}\{\cdot\}$ denotes the expectation operator of a random variable and $\text{Pr}\{ \cdot \}$ is the probability function. $\Gamma(z)$ is the gamma function 
\cite[eq.(5.2.1)]{MathFunctions2010}. 
$H^{m{,}n}_{p{,}q}\left[\cdot\right]$ and $H\left[\cdot\right]$ denote the univariate Fox's \textit{H} function or multivariate Fox's \textit{H} function defined by \eqref{eq H_function_notation}-\eqref{eq definition of product of H_functions} depending on the variates involved\footnote{For notationally simplicity, we use the similar format to represent both univariate and multivariate Fox's \textit{H} function, which can be distinguished by the number of variates involved. Since univariate Fox's \textit{H} function is included as a special case of multivariate Fox's \textit{H} function, they can be treated universally as multivariate Fox's \textit{H} function. }%
, which are detailed in Appendix \ref{subsection Notations and Symbols}. $\mathbb{H}\{\cdot\}(\cdot)$ denotes the Fox's \textit{H} transform. A bold dash ``$\textbf{---}$'' is used where no parameter is present. $f(\gamma)$ denotes PDF function, while $\phi(s)$ denotes MGF function. 
Some frequently used parameter sequences have been summarized in Table \ref{table sequences}.

\begin{table}[!ht]
\centering
\caption{Frequently Used Order and Parameter Sequences}
\label{table sequences}
\resizebox{0.48\textwidth}{!}{%
\begin{tabular}{|l|l|}
\hline
Symbols & Order and Parameter Sequences \\ \hline
\begin{tabular}[c]{@{}l@{}}$\mathbf{O}_\text{MGF}$\\ $\mathbf{P}_\text{MGF}$\end{tabular} & \begin{tabular}[c]{@{}l@{}}$(1,0,0,1)$\\ $(1,1,\textbf{---},0,\textbf{---},1)$\end{tabular} \\ \hline
\begin{tabular}[c]{@{}l@{}}$\mathbf{O}_\text{ER}$\\ $\mathbf{P}_\text{ER}$\end{tabular} & \begin{tabular}[c]{@{}l@{}}$(1,0,0,1)$\\ $(1,1,\textbf{---},A-1,\textbf{---},1)$\end{tabular} \\ \hline
\begin{tabular}[c]{@{}l@{}}$\mathbf{O}_f$\\ $\mathbf{P}_f$\end{tabular} & \begin{tabular}[c]{@{}l@{}}$(m,n,p,q)$\\ $(u,v,\mathbf{c},\mathbf{d},\mathbf{C},\mathbf{D})$\end{tabular} \\ \hline
\begin{tabular}[c]{@{}l@{}}$\mathbf{O}_{\phi}$\\ $\mathbf{P}_{\phi}$\end{tabular} & \begin{tabular}[c]{@{}l@{}}$(n+1,m,q,p+1)$\\ $(u/v,1/v,1-\mathbf{d}-\mathbf{D},(0,1-\mathbf{c}-\mathbf{C}),\mathbf{D},(1,\mathbf{C}))$\end{tabular} \\ \hline
\end{tabular}
}
\end{table}

\section{System Model}
\label{section System Model and Background Knowledge}
In this paper, a MISO fading channel model is considered, where there are $N$ transmit antennas and only one receive antenna as shown in Fig. \ref{fig_MISO}. The channels are assumed to be flat block fading, then the channel input-output relation can be written as
\begin{equation}
y=\mathbf{hx}+n_{_0},
\end{equation}
where $\mathbf{h}\in {\mathbb{C}}^{1 \times N}$ denotes the MISO channel vector, $\mathbf{x}$ is the transmit signal vector, and $n_{_0}$ represents the complex additive white Gaussian noise with zero mean and variance $N_0$. It is assumed that the transmission power is uniformly allocated across the transmit antennas, and the channels are assumed to be arbitrarily correlated and not necessarily identically distributed.
\begin{figure}[!htbp]
\centering
\includegraphics[width=0.48\textwidth]{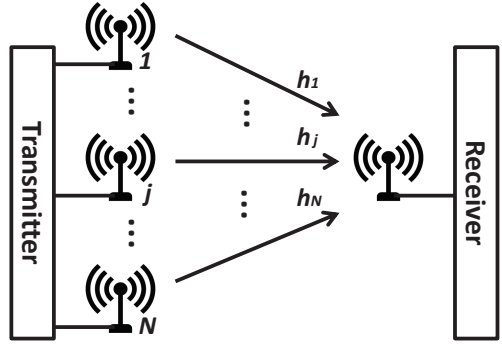}
\caption{MISO system model.}
\label{fig_MISO}
\end{figure}

The average transmit SNR is defined as $\rho=\frac{P}{BN_0}$, where $P$ is the average transmit power of the system and $B$ denotes the bandwidth. The channel state information is assumed to be only available at the receiver and the instantaneous channel power gain of the $j$th channel is defined by $\gamma_j=|h_j|^2$, where $h_j$ is the $j$th component of the fading vector $\mathbf{h}$. 
In this paper, we consider the maximum ratio transmission scheme, hence the instantaneous channel power gain at the receiver end can be defined by $\gamma_\text{end}{=}\sum _{j{=}1}^N \gamma_j$. 
The joint MGF $\phi_\text{end}$ is defined by $\phi_\text{end}(s)=\mathbb{E}\{e^{-s\sum_{j=1}^N \gamma_j}\}$. Specifically, if the channels are independent with each other, then the joint MGF $\phi_{{\text{end}}}(s)$ can be expressed by the product of the MGF of individual channel's power gain as 
\begin{equation}
\label{eq MGF equal to prod of MGF_j}
\phi_\text{end}(s)= \prod _{j{=}1}^N \phi_{j}(s),
\end{equation}
where $\phi_{j}(s)$ is defined by $\phi_{j}(s)=\int _0 ^\infty e^{-s\gamma}f_{j}(\gamma)d\gamma$ and $f_{j}(\gamma)$ is the PDF of the $j$th channel's power gain. 

\section{Effective Rate Analysis Using MGF}
\label{section effective rate general}
Effective rate is the maximum constant rate that a fading channel can support under statistical delay constraints, which can be written as \cite{EffCap}

\begin{equation}
R(\theta)=-\frac{1}{\theta T} \ln\mathbb{E}\left\{e^{-\theta T C}\right\},\quad\theta \neq 0,
\end{equation}
where $C$ represents the system's throughput during a single time block and $T$ denotes the duration of a time block. The QoS exponent $\theta$ is given by \cite{EffCap}
\begin{equation}
\theta = - \lim _{z \rightarrow \infty} \frac{\ln \text{Pr} \{ L>z \}}{z},
\end{equation}
where $L$ is the equilibrium queue-length of the buffer at the transmitter.
When $z$ is large, the buffer violation probability can be approximated by $\text{Pr}\{L\geq z\}\approx e^{-\theta z}$. Correspondingly, if we denote the steady-state delay at the buffer by $D_\text{delay}$, then the probability of $D_\text{delay}$ exceeding the maximum allowed delay $d_\text{max}$ can be given by $\text{Pr}\{D_\text{delay}\geq d_\text{max}\}\approx e^{-\theta \delta d_\text{max}}$, where $\delta$ is determined by the characters of the queueing system \cite{lowSNR2011}.
Hence the minimum required QoS exponent $\theta_0$ is decided by the delay constraints, and in order to guarantee the delay performance, the QoS exponent $\theta$ has to satisfy the constraint $\theta \geq \theta _0$. Moreover, when $\theta _0 \rightarrow 0$, the effective rate approaches the Shannon's capacity \cite{lowSNR2011}.

When the transmitter sends uncorrelated circularly symmetric zero-mean complex Gaussian signals, the effective rate can be written as \cite{lowSNR2011}
%
\begin{equation}
\label{eq ER_theory}
R(\theta){=}{-}\frac{1}{A} \log_2 \mathbb{E} \left\{ \left(1{+}\frac{\rho}{N} \gamma_\text{end}  \right)^{{-}A} \right\},
\end{equation}
where $A=\theta T B \slash \ln2$. 
Normally, \eqref{eq ER_theory} will involve multiple integrations of multivariate functions, since the PDF of $\gamma_\text{end}=\sum_{j=1}^N \gamma_j$ is described by multivariate random variables in general cases. Yet by using the MGF instead of PDF and applying the \textit{H} transform theory, the effective rate can be derived as follows. 
\begin{theorem}
\label{theorem MISO_MGF_H-transform}
The effective rate over arbitrary correlated and not necessarily identically distributed MISO fading channels can be given by
\begin{align}
R(\theta)
\label{eq ER MGF H transform form s}
{=}&{-}\frac{1}{A} \log_2 \frac{1}{\Gamma(A)} \mathbb{H} \left\{\phi_\text{end}(\frac{\rho s}{N}),\mathbf{O}_\text{ER},\mathbf{P}_\text{ER} \right\}\big(1\big),\\
\label{eq ER MGF H transform form identical}
{\equiv}&{-}\frac{1}{A} \log_2 \frac{N}{\rho\Gamma(A)} \mathbb{H} \left\{\phi_\text{end}(s),\mathbf{O}_\text{ER},\mathbf{P}_\text{ER} \right\}\big(\frac{N}{\rho}\big),
\end{align}
where $\mathbf{O}_\text{ER} {=} (1,0,0,1)$ and $\mathbf{P}_\text{ER}{=}(1{,}{1}{,}{\textbf{---}}{,}A{-}1{,}{\textbf{---}}{,}1)$. 
\end{theorem}
\begin{proof}
See Appendix \ref{app proof MISO_MGF_H-transform}.
\end{proof}

\emph{Remark}:
The integral form of the effective rate of MISO fading channels can be obtained by using \textit{H} transform definition and substituting the identity of $e^{{-}s}s^{A{-}1}=H^{1,0}_{0,1}\left[s\bigg|\begin{matrix}
{\textbf{---}}\\
(A{-}1,1)
\end{matrix} \right]$ \cite[eq.(2.22)]{HFunction_thebook2010} into \eqref{eq ER MGF H transform form s}, which can be written as
\begin{equation}
\label{eq numerical_MGF_expression_s}
R(\theta){=}{-}\frac{1}{A} \log_2 \frac{1}{\Gamma(A)}\int_0^\infty e^{{-}s} s^{A{-}1}\phi_\text{end}(\frac{\rho s}{N})ds.
\end{equation}
Using the fact that $0\leq\phi_\text{end}(s)\leq 1$, the integral can be proved to be absolutely converge. Hence the \textit{H} transform in \eqref{eq ER MGF H transform form s} and \eqref{eq ER MGF H transform form identical} exist.

We highlight that \eqref{eq ER MGF H transform form s} and its equivalent \eqref{eq ER MGF H transform form identical} in \textit{H} transform form are more attractive than the integral form in \eqref{eq numerical_MGF_expression_s}.
On one hand, the MGF of most commonly used fading distributions can be written in Fox's \textit{H} function format\cite{HyperFoxHFading2012}, which facilitates the application of Theorem \ref{theorem MISO_MGF_H-transform}.
On the other hand, we can interpret the effective rate as the \textit{H} transform of the PDF or MGF of the power gain by the parameter sequence $\mathbf{O}_\text{ER}$ and $\mathbf{P}_\text{ER}$ with some more manipulations. It has been studied in \cite{H_Transform2015} that many important metrics such as ergodic capacity, error probability, error exponent can be expressed in similar \textit{H} transform format with corresponding parameter sequences. Note that one merit of \textit{H} transform is that the manipulation of parameter sequences only involves very basic arithmetic operations, hence the \textit{H} transform representation can provide a unified, systematic and simple framework for wireless performance analysis such as effective rate. In addition, in Section \ref{section ER H and hyper H}, we will show that for deriving MGF from the known PDF or analysing effective rate over SISO, i.n.i.d. and correlated scenarios, \textit{H} transform is a very useful analysing tool.

In addition, compared to the PDF based approach, the MGF based approach has many advantages. 
First, the PDF based approach can be viewed as a special case of the MGF based approach. This is due to the fact that in the PDF based approach, the joint PDF has to be obtained in advance of further analysis. When the joint PDF is available, by using the relationship of PDF and MGF, it can be proved that the PDF based approach is involved in the case of $N=1$ in the proposed MGF based approach. 
Second, in the i.n.i.d. scenarios, the joint MGF can be calculated by the product of the individual channels' MGF, which enables the analyses where the PDF based approach has difficulties \cite{kappa_mu_inid_2014, alpha_mu_EffRate_2015}. This feature makes the analysis more flexible and applicable to different and complex channel conditions.

%
%
%
\section{Exact effective rate over MISO fading channels}
\label{section ER H and hyper H}
In this section, we will investigate the effective rate over specific i.n.i.d. and correlated channels. For the i.n.i.d. scenario, the effective rate over i.n.i.d. hyper Fox's \textit{H} fading channels is derived, while the arbitrary correlated generalized $K$ fading channels are considered for the correlated scenario.
\subsection{Effective rate over i.n.i.d. MISO hyper Fox's \textit{H} fading channels}
The hyper Fox's \textit{H} fading model \cite{HyperFoxHFading2012} 
uses the sum of several \textit{H} variates to exactly represent or approximate a very wide range of different fading distribution models, including Rayleigh, Weibull, Nakagami-$m$, Weibull/Gamma and generalized $K$ fading models. For a full list of special cases, interested readers are referred to \cite{HyperFoxHFading2012}. By matching the parameters, effective rates over various i.n.i.d. MISO fading channels are readily available, which will further simplify the effective rate analysis under such conditions. 

Let $\gamma_j$ be a random variable following the hyper Fox's \textit{H} fading distribution, then its PDF is given by \cite{HyperFoxHFading2012}
\begin{align}
\label{eq Hyper Fox's H Fading or notation}
f_{j}(\gamma)
=\sum_{k{=}1}^{K_j}H[\gamma,\mathbf{O}_{f_{_{j,k}}},\mathbf{P}_{f_{_{j,k}}}],
\end{align}
where 
\begin{equation}
\label{eq O and P gamma jr definition}
\left\lbrace \begin{aligned}
&\mathbf{O}_{f_{j,k}}{=}(m_{_{j,k}},n_{_{j,k}},p_{_{j,k}},q_{_{j,k}})\\
&\mathbf{P}_{f_{_{j,k}}}{=}(u_{_{j,k}},v_{_{j,k}},\mathbf{c}_{_{j,k}}{,}\mathbf{d}_{_{j,k}}{,}
\mathbf{C}_{_{j,k}}{,}\mathbf{D}_{_{j,k}})
\end{aligned}\right.
\end{equation}
defined over $\gamma \geq 0$ and the subscripts of $m_{_{j,k}}$ denote that this parameter is associated to the $k$th parameter set of the PDF $f_j(\gamma)$, where $k\in (1,2,\dots,K_j)$. Same notation rule applies to the other parameters. 
The notations $K_j$, $\mathbf{O}_{f_{_{j,k}}}$and $\mathbf{P}_{f_{_{j,k}}}$are the parameters satisfying a distributional structure such that $f_{j}(\gamma)\geq 0$ for all $\gamma\geq 0$ and $\int_0^\infty f_{j}(\gamma) d\gamma=1$ \cite{HyperFoxHFading2012}. The necessary condition for the Fox's \textit{H} function in \eqref{eq Hyper Fox's H Fading or notation} to be a density function can be given by \cite{H_Transform2015,Cook_H_thesis1992,HVariates1972}
\begin{equation}
\left\lbrace 
\begin{aligned}
\label{eq Hyper H necessary conditions}
c_j+C_j<1, \forall j=1,2,\dots,n\\
-\frac{d_j}{D_j}<1, \forall j=1,2,\dots,m
\end{aligned}\right.
\end{equation}

The MGF $\phi_{j}$ can be obtained by using \textit{H} transform defined in \eqref{eq H Transform definition} as \cite{H_Transform2015}
\begin{equation}
\label{eq hyper Fox's H MGF definition}
\begin{aligned}
\phi_{j}(s)
{=}\sum_{k{=}1}^{K_j} H[s,\mathbf{O}_{\phi_{_{j,k}}},\mathbf{P}_{\phi_{_{j,k}}}],\\
\end{aligned}
\end{equation}
where the parameter sequences of the MGF can be calculated by
\begin{equation}
\label{eq O P phi jr hyper Fox's H definition}
\left\lbrace
\begin{aligned}
&\mathbf{O}_{\phi_{_{j,k}}}{=}\mathbf{O}_\text{MGF} \boxdot \mathbf{O}_{f_{_{j,k}}}\\
&\mathbf{P}_{\phi_{_{j,k}}}{=}\mathbf{P}_\text{MGF} \boxdot \mathbf{P}_{f_{_{j,k}}}
\end{aligned}\right.
\end{equation}
where $\mathbf{O}_\text{MGF} = (1,0,0,1)$ and
$\mathbf{P}_\text{MGF} =(1,1,{\textbf{---}},0,{\textbf{---}},1)$ as given in Table \ref{table sequences}. The Mellin operation $\boxdot$ is defined in Appendix \ref{subsection Notations and Symbols}, which is a typical operation in \textit{H} transform. Mellin operation is very useful since when the integration kernel is fixed, for example in the case of deriving MGF from PDF, it uses basic arithmetic manipulation of parameters to replace the integration operation procedure.

As the derivation will involve the \textit{H} transform of multivariate \textit{H} function, we introduce the following lemma, which enables us to analyse the effective rate over i.n.i.d. MISO fading channels.
\begin{lemma}
\label{theorem Laplace transform of the multiplication of the H functions}
The \textit{H} transform of the product of Fox's \textit{H} functions can be given by
\begin{equation}
\begin{aligned}
&\mathbb{H}\left\lbrace \prod _{j{=}1}^N H\left[z,\mathbf{O}_j,\mathbf{P}_j\right],\mathbf{O}_\text{R'}, \mathbf{P}_\text{R'}
\right\rbrace(s)\\
\label{eq Laplace transform of the multiplication of the H functions multivariate H form}
{=}& s^{{-}\eta}H^{0{,}1}_{1{,}0}\left[\begin{matrix}
(1{-}\eta,\mathbf{1}_N)\\
{\textbf{---}}
\end{matrix}{:}\left(\frac{1}{s}{,}\mathbf{O}_j{,}\mathbf{P}_j\right)_{1{,}N}\right],
\end{aligned}
\end{equation}
where $\mathbf{O}_\text{R'}=(1,0,0,1)$, $\mathbf{P}_\text{R'}=(s^{1-\eta},1,\textbf{---},\eta-1,\textbf{---},1)$, $s>0$ and $\eta \ge 0$. The short notation of multivariate Fox's \textit{H} function $H^{0{,}n_0}_{p_0{,}q_0}\left[\cdot\right]$ is defined by \eqref{eq multivariate_H_function_definition} in Appendix \ref{subsection Notations and Symbols}.
Specially, when $N=1$, \eqref{eq Laplace transform of the multiplication of the H functions multivariate H form} reduces to 
\begin{equation}
\begin{aligned}
\label{eq Laplace transform of single H functions}
&\mathbb{H}\left\lbrace uH^{m{,}n}_{p{,}q}\left[vz{\bigg|}
\begin{matrix}
 \mathbf{c}{,}\mathbf{C}\\
 \mathbf{d},\mathbf{D}
\end{matrix}
\right],\mathbf{O}_\text{R'}, \mathbf{P}_\text{R'}
\right\rbrace(s)\\
{=}&us^{{-}\eta} H^{m{,}n+1}_{p+1{,}q}\!\left[\frac{v}{s}{\bigg|}\begin{matrix}
 (1{-}\eta{,}\mathbf{c}){,}(1{,}\mathbf{C})\\
 \quad\,\,\mathbf{d},\mathbf{D}
\end{matrix} \right].
\end{aligned}
\end{equation}
\end{lemma}
\begin{proof}
Using \cite[eq.(2.1)]{forLemma11986} and the \textit{H} transform definition \eqref{eq H Transform definition},  \eqref{eq Laplace transform of the multiplication of the H functions multivariate H form} can be obtained. Applying \eqref{eq H_function_notation}, \eqref{eq Laplace transform of single H functions} can be obtained.
\end{proof}
It is straightforward that the \textit{H} transform parameter $\mathbf{O}_\text{ER}$ and $\mathbf{P}_\text{ER}$ in Theorem \ref{theorem MISO_MGF_H-transform} are special cases of $\mathbf{O}_\text{R'}$ and $\mathbf{P}_\text{R'}$ in Lemma \ref{theorem Laplace transform of the multiplication of the H functions}. Hence with the MGF representation of $\gamma_j$ and Lemma \ref{theorem Laplace transform of the multiplication of the H functions}, the effective rate over i.n.i.d. MISO hyper Fox's \textit{H} fading channel can be given by the following theorem.
\begin{theorem}
\label{theorem miso ER hyper Fox's H fading}
If $N$ channels of the MISO systems are mutually independent but not necessarily identical distributed and the instantaneous channel power gain $\gamma_j$ of each channel follows hyper Fox's \textit{H} fading, then the effective rate over i.n.i.d. MISO hyper Fox's \textit{H} fading channels can be written as
\begin{equation}
\label{eq miso ER hyper Fox's H fading}
\begin{aligned}
R(\theta){=}{-}\frac{1}{A} \log_2\frac{1}{\Gamma(A)}\sum_{k_1{=}1}^{K_1}\cdots\sum_{k_N{=}1}^{K_N} H^{0{,}1}_{1{,}0}\left[\begin{matrix}
(1{-}A,\mathbf{1}_N)\\
{\textbf{---}}
\end{matrix}
{:}\right.\\
\left.\left( \frac{\rho}{N}{,}\mathbf{O}_{\phi_{j,k_j}},\mathbf{P}_{\phi_{j,k_j}}\right)_{1{,}N}\right].
\end{aligned}
\end{equation}
\end{theorem}
\begin{proof}
Substituting both \eqref{eq hyper Fox's H MGF definition} and \eqref{eq MGF equal to prod of MGF_j} into Theorem \ref{theorem MISO_MGF_H-transform}, then using Lemma \ref{theorem Laplace transform of the multiplication of the H functions}, \eqref{eq miso ER hyper Fox's H fading} can be obtained.
\end{proof}

\emph{Remark:}
Theorem \ref{theorem miso ER hyper Fox's H fading} is a good example about how to apply Theorem \ref{theorem MISO_MGF_H-transform} to estimate the i.n.i.d. fading channels in a general way. As will be illustrated in \ref{subsection special cases}, the equation \eqref{eq miso ER hyper Fox's H fading} can be much simplified in some special cases and estimated thereafter. However, as presented in Section \ref{section approximation}, the equation \eqref{eq miso ER hyper Fox's H fading} using multivariate Fox's \textit{H} function can be evaluated uniformly without the need of further reduction or simplification, which simplify the analysis and calculation procedure in a general way.

Although many parameters have been used to describe the hyper Fox's \textit{H} fading model, according to \cite{HyperFoxHFading2012}, most commonly used fading channels have very simple parameters. 
Another interesting observation is that the parameter sequences in \eqref{eq miso ER hyper Fox's H fading} are the parameter sequences of the involved MGF functions without changes. By substituting corresponding parameters in specific channel scenarios, Theorem \ref{theorem miso ER hyper Fox's H fading} is directly applicable to the analysis of effective rate over various i.n.i.d. fading channel conditions.
In addition, the multivariate Fox's \textit{H} function is a mathmatical tracable function. There are researches on the property\cite{ref_multiH_property_1977}, reduction\cite{ref_Hfunction_expansion2013}, expansion\cite{ref_expansion_multiH1976} and integrations\cite{ref_integration_multiH1983} involving multivariate Fox's \textit{H} functions, which can be used for the simplification in special cases and derivation in applications. Interested readers can refer to the references therein for further details. 
\subsection{Special cases of i.n.i.d. hyper Fox's \textit{H} fading channels}
\label{subsection special cases}
We now investigate the effective rate over three special fading channels, namely i.n.i.d. Fox's \textit{H} fading channels, i.i.d. Nakagami-$m$ fading channels and SISO hyper Fox's \textit{H} channel. On one hand, the scenarios considered in these special cases are very practical and widely used in the study of wireless system performance. On the other, these special cases give examples of how to apply the proposed theorems under specific channel conditions. 
\subsubsection{i.n.i.d. Fox's \textit{H} fading channels}\hfill\\
\indent Fox's \textit{H} fading model can be included as a special case in the hyper Fox's \textit{H} fading model, which can characterize the fluctuations of the signal envelope due to  multipath fading superimposed on shadowing variations\cite{H_Transform2015}. Typical fading models describing both small-scale and large-scale fading effects, such as Rayleigh/lognormal, Nakagami-$m$/lognormal and Weibull/Gamma fading, are all special cases of the Fox's \textit{H} fading model. This model is extensively used in the research of fading channels due to its clear physical meaning\cite{H_Transform2015,ref_HFading1977}. Every fading effect in the Fox's \textit{H} fading is defined based on the Fox's \textit{H} variate, which uses the Fox's \textit{H} function to describe the PDF of the random variate. Hence when defining a Fox's \textit{H} variate, it only needs to define its parameter sequences in the format of $X\sim H(\mathbf{O},\mathbf{P})$. For simplicity, the Fox's \textit{H} function, Fox's \textit{H} transform and their associated operation functions are exploited, which have been detailed in Appendix \ref{subsection Notations and Symbols}. These operations are all basic manipulations of the parameter sequences, which is beneficial in compositing different fading effects as well as deriving MGF from PDF as shown below.   

For a single channel labelled with $j\,(j{=}1{,}2{,}\dots{,}N)$, let the non-negative random variable $Z_j$ and  $X_j$ be Fox's \textit{H} variates\cite{ref_HFading1977} and describe the multipath fading effect and shadowing effect, such that $Z_j\sim H(\mathbf{O}_{_{Z_j}},\mathbf{P}_{_{Z_j}})$ and $X_j\sim H(\mathbf{O}_{_{X_j}},\mathbf{P}_{_{X_j}})$ \cite{H_Transform2015}.

If the multipath fading and shadowing effects are statistically independent,
then the instantaneous channel power gain $\gamma_j$ is again Fox's \textit{H} variate $\gamma_j \sim H(\mathbf{O}_{f_j}{,}\mathbf{P}_{f_j})$,
where $\mathbf{O}_{f_j}{=} \mathbf{O}_{Z_j}\boxplus \mathbf{O}_{X_j}$ and $\mathbf{P}_{f_j}{=} \langle1{,}2{,}{{-}\frac{1}{2}}{\big|} \mathbf{P}_{Z_j}\boxplus \mathbf{P}_{X_j}$. The elementary operation $\langle \cdot | \mathbf{P}$ and convolution operation $\boxplus$ are defined in Appendix \ref{subsection Notations and Symbols}.

The MGF of the instantaneous channel power gain $\gamma_j$ can be written as 
\begin{equation}
\label{eq siso MGF definition}
\phi_{j}(s){=}H[s{,\mathbf{O}_{\phi_{j}},\mathbf{P}_{\phi_{j}}}],
\end{equation} 
with
\begin{equation}
\label{eq O phi gamma and P phi gamma definition}
\left\lbrace
\begin{aligned}
&\mathbf{O}_{\phi_{j}}=&\mathbf{O}_{\text{MGF}}\boxdot\mathbf{O}_{f_j}\hfill \\
&\mathbf{P}_{\phi_{j}}=&\mathbf{P}_{\text{MGF}}\boxdot\mathbf{P}_{f_j} \hfill
\end{aligned}\right.
\end{equation}
where Mellin operation $\boxdot$ is defined in Appendix \ref{subsection Notations and Symbols}.
Note the MGF of Fox's \textit{H} fading model is represented by only one Fox's \textit{H} function, it coincides as a special case of hyper Fox's \textit{H} fading model with the parameter $K_j=1$ in \eqref{eq Hyper Fox's H Fading or notation}. Then the effective rate over i.n.i.d. MISO Fox's \textit{H} fading channels can be given by 
\begin{equation}
\label{eq miso ER H-Fading equation}
\begin{aligned}
R(\theta)
{=}{-}\frac{1}{A} \log_2\frac{1}{\Gamma(A)} H^{0{,}1}_{1{,}0}\left[\begin{matrix}
(1{-}A,\mathbf{1}_N)\\
{\textbf{---}}
\end{matrix}
{:}\left( \frac{\rho}{N}{,}\mathbf{O}_{\phi_{j}},\mathbf{P}_{\phi_{j}}\right)_{1{,}N}\right]
\end{aligned}
\end{equation}
where $\mathbf{O}_{\phi_{j}}$ and  $\mathbf{P}_{\phi_{j}}$ are defined in \eqref{eq O phi gamma and P phi gamma definition}.

This special case gives a good example of how to derive the effective rate from the known fading parameters of the specific fading channel via the proposed method and \textit{H} transform operations. It should also be noticed that although several \textit{H} transform operations are used in the deriving procedure, they only involve some basic arithmetic manipulations of the parameters, such as addition, subtraction, multiplication, division, sequence changes of the parameters as well as the combination of such operations. One can readily obtain the more familiar yet tedious representation by expanding these operations described in Appendix \ref{subsection Notations and Symbols}.


\subsubsection{I.i.d. Nakagami-$m$ fading channels}\hfill\\
\indent Nakagami-$m$ fading model is one of the most widely used fading model in the performance analysis of wireless communication systems, which includes one-sided Gaussian and Rayleigh fading model as special cases. It has been recommended by IEEE Vehicular Technology Society Committee on Radio Propagation for theoretical studies of fading channels \cite{VTC_recommend_Weibull1988}. For a single channel, if the instantaneous channel power gain follows Nakagami-$m$ distribution, then $\gamma_j\sim H(\mathbf{O}_j,\mathbf{P}_j)$, where $\mathbf{O}_j=(1,0,0,1)$ and $\mathbf{P}_j=(\frac{{\hat{m}}}{\Gamma(\hat{m})},{\hat{m}},{\textbf{---}},\hat{m}{-}1,{\textbf{---}},1)$\cite{H_Transform2015}. The symbol $\hat{m}$ denotes the parameter associated to  Nakagami-$m$ fading model. Using \eqref{eq siso MGF definition}, the MGF can be given by 
\begin{equation}
\label{eq eg siso MGF nakagami m}
\phi_{j}(s){=}H[s{,\mathbf{O}_{\phi_{j}},\mathbf{P}_{\phi_{j}}}],
\end{equation} 
where  
$\mathbf{O}_{\phi_{j}}{=}(1,1,1,1)$ and 
$\mathbf{P}_{\phi_{j}}{=}(\frac{1}{\Gamma(\hat{m})},\frac{1}{\hat{m}},1{-}\hat{m},0,1,1)$. 
Substituting these parameters into Theorem \ref{theorem miso ER hyper Fox's H fading} as well as the relation connecting generalized Lauricella function and multivariate Fox's \textit{H} function \cite[A.31]{HFunction_thebook2010}, then using the reduction formulae for the multivariate hypergeometric function \cite[eq.(14)]{ref_reduction_hypergeometric1985}, the effective rate expression can be simply represented by generalized hypergeometric functions ${}_pF_q[\cdot]$ \cite[eq.(16.2.1)]{MathFunctions2010} as follows
\begin{equation}
\label{eq eg iid nakagami}
\begin{aligned}
R(\theta)&=-\frac{1}{A}\log_2 \,{}_2F_0\left[A,\hat{m}N,-\frac{\rho}{\hat{m}N} \right].\hfill
\end{aligned}
\end{equation}
Note that using the identity of \cite[eq.(13.6.21)]{MathFunctions2010}, the expression given in \eqref{eq eg iid nakagami} coincides with \cite[eq.(7)]{MISOFramework2012}, which supports the validation of our derivation. 
 

\subsubsection{SISO hyper Fox's \textit{H} fading channel}\hfill\\
\indent SISO fading channel is included as a special case of i.n.i.d. channel, where only one channel is considered, i.e. $N=1$. Under SISO hyper Fox's \textit{H} fading channel condition, the MGF can be expressed by the sum of univariate Fox's \textit{H} functions as $\phi_{{\text{end}}}(s)=\sum_{k{=}1}^{K} H[s,\mathbf{O}_{\phi_{k}},\mathbf{P}_{\phi_{k}}]$, where $\mathbf{O}_{\phi_{k}}=(m_{\phi_{k}}{,}n_{\phi_{k}}{,}p_{\phi_{k}}{,}q_{\phi_{k}})$ and $\mathbf{P}_{\phi_{k}}{=}(u_{\phi_{k}}{,}v_{\phi_{k}}{,}\mathbf{c}_{\phi_{k}}{,}\mathbf{d}_{\phi_{k}}{,}\mathbf{C}_{\phi_{k}}{,}\mathbf{D}_{\phi_{k}})$. In this case, using \eqref{eq ER MGF H transform form s} in Theorem \ref{theorem Laplace transform of the multiplication of the H functions} and \eqref{eq Laplace transform of single H functions} in Lemma \ref{theorem Laplace transform of the multiplication of the H functions}, the effective rate expression can be directly given by
\begin{equation}
\label{eq ER N=1 hyper Fox's H fading}
\begin{aligned}
&R(\theta){=}\\
&{-}\frac{1}{A} \log_2\sum_{k{=}1}^{K}\frac{u_{\phi_{k}}}{\Gamma(A)}H^{m_{\phi_{k}}{,}n_{\phi_{k}}+1}_{p_{\phi_{k}}+1{,}q_{\phi_{k}}}\left[\frac{v_{\phi_{k}}}{\rho}{\bigg|}\begin{matrix}
 (1{-}A{,}\mathbf{c}_{\phi_{k}}){,}(1{,}\mathbf{C}_{\phi_{k}}) \\
 \mathbf{d}_{\phi_{k}},\mathbf{D}_{\phi_{k}}
\end{matrix} \right].
\end{aligned}
\end{equation}

Note that \eqref{eq ER N=1 hyper Fox's H fading} can be also obtained by substituting $N=1$ into \eqref{eq miso ER hyper Fox's H fading}. It can be verified that  \eqref{eq ER N=1 hyper Fox's H fading} coincides with \cite[eq.(7)]{ER_SISO_MGF_H_2014}, which supports the validation of the derivation. 
%
Also in this special case, the effective rate representation in \eqref{eq ER N=1 hyper Fox's H fading} only involves univariate Fox's \textit{H} functions. 
%

\subsection{Effective rate over arbitrary correlated generalized $K$ fading channels}
If the transmit antennas are sufficiently separated in space, it is reasonable to assume independence between the received signals from different antennas. Yet this assumption may be crude for some systems, where correlation between antennas is a more practical scenario. Hence in this part, we investigate the effective rate over the arbitrary correlated generalized $K$ fading channels. 

Generalized $K$ fading model has been first introduced in \cite{generalized_K_fading1986} to model the intensity of radiatio scattered with a non-uniform phase distribution (weak-scatterer regime), which accounts for the composite effect of Nakagami-$m$ multipath fading and gamma shadowing. 
%
%
%
Here we assume the multipath fading effect and shadowing effect are independent with each other, which are represented by random variable $\omega_j$ and $\xi_j$, respectively. Let $\gamma_j=\xi_j \omega_j$ with $j=1,2,\dots,N$ and assume $\xi_j(j=1,2,\dots,N)$ are i.i.d. Gamma-distributed random variables with parameter $m_1$, while $\omega_j(j=1,2,\dots,N)$ are identically distributed with arbitrary correlation matrix $\mathbf{\Sigma}$ and parameter $m_2$. The elements of the correlation matrix are given by $\Sigma_{i,j}=1$ for $i=j$ and $\Sigma_{i,j}=r_{i,j}$ for $i\neq j$, where $0\leq r_{i,j} < 1$ is the correlation coefficient between channel $i$ and $j$. Using \cite[eq.(9)]{ref_correlated_GammaGamma} as well as the relation between Meijer's \textit{G} function and Fox's \textit{H} function \cite[eq.(1.112)]{HFunction_thebook2010}, the joint MGF of $\mathbf{{\gamma\mkern-11mu\gamma}}=[\gamma_1,\gamma_2,\dots,\gamma_N]$ can be given by
\begin{equation}
\label{eq MGF correlated GammaGamma}
\begin{aligned}
\phi_{\gamma\mkern-11mu\gamma}(s)=&\frac{\text{det}|\mathbf{W}|^{m_1}}{\Gamma(m_1)[\Gamma(m_2)]^N} \sum_{k_1,\dots,k_{N{-}1}{=}0}^{\infty}\left(\prod^{N-1}_{l=1}\frac{|p_{l,l+1}|^{2k_l}}{k_l!\Gamma(m_1{+}k_l)}\right.\\
&\times {\prod_{j=1}^{N}\mathbf{H}[s,\mathbf{O}_j,\mathbf{P}_j]}\Bigg),
\end{aligned}
\end{equation}
where $\mathbf{O}_j=(1,2,2,1)$ and $\mathbf{P}_j=(p_{j,j}^{-m_1-\alpha_j},\frac{1}{p_{j,j}m_1m_2},(1-m_1-\alpha_j,1-m_2),0,(1,1),1)$ with $\alpha_j=k_1$ for $j=1$,  $\alpha_j=k_{N-1}$ for $j=N$ and  $\alpha_j=k_{j-1}+k_j$ for $j=2,3,\dots,N-1$. $\mathbf{W}$ is the inverse of $\mathbf{\Sigma}$, whose elements are denoted by $p_{i,j}$. By substituting \eqref{eq MGF correlated GammaGamma} into Theorem \ref{theorem MISO_MGF_H-transform} and applying Lemma \ref{lemma T approximation}, the effective rate of generalized $K$ fading channels with arbitrary correlation matrix can be given by
\begin{align}
\label{eq ER correlated GammaGamma}
R(\theta)=&-\frac{1}{A}\log_2\frac{\text{det}|\mathbf{W}|^{m_1}}{\Gamma(A)\Gamma(m_1)[\Gamma(m_2)]^N}\times\\ \notag
& \sum_{k_1,k_2,\dots,k_{N-1}=0}^{\infty}\left(\prod^{N-1}_{l=1}\frac{|p_{l,l+1}|^{2k_l}}{k_l!\Gamma(m_1+k_l)}\right.\times\\ \notag
&H^{0{,}1}_{1{,}0}\Bigg[\begin{matrix}
(1{-}A,\mathbf{1}_N)\\
{\textbf{---}}
\end{matrix}
{:}\left.\left( \frac{\rho}{N}{,}\mathbf{O}_{{j}},\mathbf{P}_{{j}}\right)_{1{,}N}\Bigg]\right).
\end{align}
Note that if $\gamma_j$ are independent with each other, then $\mathbf{W}=\mathbf{I}$ and \eqref{eq ER correlated GammaGamma} reduces to the case of i.i.d. generalized $K$ fading channels as
\begin{equation}
\label{eq ER i.i.d. Generalized-K}
R(\theta)=-\frac{1}{A}\log_2 \frac{H^{0{,}1}_{1{,}0}\left[\begin{matrix}
(1{-}A,\mathbf{1}_N)\\
{\textbf{---}}
\end{matrix}
{:}\left( \frac{\rho}{N}{,}\mathbf{O}_{{j}},\mathbf{P}_{{j}}\right)_{1{,}N}\right]}{\Gamma(A)[\Gamma(m_1)\Gamma(m_2)]^N}, 
\end{equation}
where the operator sequence reduces to $\mathbf{O}_j=(1,2,2,1)$ and $\mathbf{P}_j=(1,\frac{1}{m_1m_2},(1-m_1,1-m_2),0,(1,1),1)$.

\section{Approximations}
\label{section approximation}
By substituting the parameters corresponding to the individual channels, we can obtain the analytical expression for the effective rate over hyper Fox's \textit{H} fading channels through Theorem \ref{theorem miso ER hyper Fox's H fading}. 
In this section, we propose a simple asymptotic approximation method, which provides an easy way to evaluate the representations presented in Section \ref{section ER H and hyper H}.

We commence on the following lemma.
\begin{lemma}
\label{lemma T approximation}
If let $E(Q)=\int _Q^\infty f(z)dz$ denote the truncation error and $E(M)=\int _0^Q f(z)dz-\sum _{\ell =0}^M f(\ell)$ denote the discretization error, then for  $0{\ll} Q{\ll} M$ the following relation can be obtained
\begin{equation}
\label{eq integration_approximation_ER}
\begin{aligned}
& \mathbb{H} \left\{\phi_{{\text{end}}}(\frac{\rho s}{N}),\mathbf{O}_\text{ER},\mathbf{P}_\text{ER} \right\}\big(1\big)
 \equiv \int_0^\infty e^{{-}s} s^{A{-}1}\phi(\frac{\rho s}{N})ds \\
= &\left(\frac{Q}{M}\right)^A\sum_{\ell{=}1}^M {\ell}^{A{-}1} e^{{-}\frac{\ell Q}{M}}\phi\left(\frac{\ell Q\rho}{MN}\right)+E(Q)+E(M),
\end{aligned}
\end{equation} 
where  $\mathbf{O}_\text{ER} {=} (1,0,0,1)$ and $\mathbf{P}_\text{ER}{=}(1{,}{1}{,}{\textbf{---}}{,}A{-}1{,}{\textbf{---}}{,}1)$ are effective rate parameter sequences defined in Theorem \ref{theorem MISO_MGF_H-transform}.
The truncation error $E(Q)$ and discretization error $E(M)$ can be given by
\begin{align}
\label{eq truncation_error EQ}
E(Q)&=\int_Q^\infty e^{{-}s} s^{A{-}1}\phi(\frac{\rho s}{N})ds,\hfill\\
\label{eq discretization_error EM}
E(M)&=-\frac{Q^3}{12M^2} \frac{\partial^2 e^{{-}s}s^{A{-}1}\phi(\frac{\rho s}{N})}{\partial s^2}{\bigg|}_{s=\xi},\quad0{<}\xi{<}Q,
\end{align}
and 
\begin{align}
\label{eq EQ=0}
\lim _{Q \rightarrow \infty } E(Q) = 0,
\lim _{M \rightarrow \infty } E(M) = 0.
\end{align}
\end{lemma}
\begin{proof}
See Appendix \ref{app proof T approximation}.
\end{proof}
By applying Lemma \ref{lemma T approximation}, the multivariate Fox's \textit{H} function involved in the effective rate calculation can be approximated using the following theorem.
\begin{theorem}
\label{theorem multivariate H approximation} 
The special series multivariate Fox's \textit{H} function can be approximated by
\begin{equation}
\label{eq multivariate H approximation}
\begin{aligned}
&H^{0{,}1}_{1{,}0}\left[\begin{matrix}
(1{-}\eta,\mathbf{1}_N)\\
{\textbf{---}}
\end{matrix}{:}\left(\frac{\rho}{N}{,}\mathbf{O}_j{,}\mathbf{P}_j\right)_{1{,}N}\right]\\
{\approx}&\left(\frac{Q}{M}\right)^\eta \sum_{\ell{=}1}^M  {\ell}^{\eta{-}1} e^{{-}\frac{\ell Q}{M}}\prod_{j=1}^{N}H\left[\frac{\ell Q \rho}{MN}{,}\mathbf{O}_j{,}\mathbf{P}_j\right],
\end{aligned}
\end{equation}
providing that each $H\left[s{,}\mathbf{O}_j{,}\mathbf{P}_j\right]$ satisfies the convergence conditions of MGF functions, $\eta \ge 0$ and  $0{\ll} Q{\ll} M$.
\end{theorem}
\begin{proof}
Applying Lemma \ref{lemma T approximation} to Lemma \ref{theorem Laplace transform of the multiplication of the H functions}, \eqref{eq multivariate H approximation} can be obtained.
\end{proof}
\emph{Remark:}
Theorem \ref{theorem multivariate H approximation} shows that the special multivariate Fox's \textit{H} function can be evaluated by the product of univariate Fox's \textit{H} functions, where there are already methods for the evaluation of univariate Fox's \textit{H} function in numerical software like Matlab. It should be noticed that although we limit the conditions of each univariate Fox's \textit{H} function involved in \eqref{eq multivariate H approximation} to be satisfying the MGF functions' convergence conditions, in fact the conditions can be further relaxed, for example when each univariate Fox's \textit{H} function approaches 0 as ${s\rightarrow \infty}$ and the second partials are bounded for all non-negative real number $s$. 

Hence in this way, we can get the general approximation formula for the MGF based approach in the following theorem. 
\begin{theorem}
\label{theorem trapezodial_approx ER}
The effective rate over arbitrary correlated and not necessarily identical MISO fading channels can be approximated by finite summation of the MGF as the following equation 
\begin{equation}
\label{eq trapezodial_approx_overall}
R(\theta){\approx}\frac{\log_2\Gamma(A)}{A}{-}\frac{1}{A}\log_2\left(\frac{Q}{M}\right)^{\!A}\!\sum^M_{\ell{=}1}  \ell^{A{-}1}e^{{-}\frac{\ell Q}{M}}\phi_{{\text{end}}}(\frac{\ell Q \rho}{MN}),
\end{equation}
where $0{\ll} Q{\ll} M$ and
the truncation error as well as the discretization error is given by \eqref{eq truncation_error EQ} and \eqref{eq discretization_error EM}.
\end{theorem}
\begin{proof}
Substituting \eqref{eq integration_approximation_ER} into \eqref{eq numerical_MGF_expression_s} and after some simple algebra manipulation, the desired result can be obtained.
\end{proof}
Compared to the exact representation (6) proposed in Theorem 1, the approximated representation (31) in Theorem 4 is also very attractive.
This representation only involves the summation of finite terms, which may help to reduce the computational complexity at the cost of accuracy. But as shown in Section \ref{section simulation}, the asymptotic approximation converges quickly since the exponential fading form is involved, where $Q=15$ and $\frac{M}{Q}=300$ provides good fit to both analytical and simulation results.

\emph{Remark:}
Applying Theorem \ref{theorem multivariate H approximation}, the effective rate of i.n.i.d. MISO hyper Fox's \textit{H} fading channels given by \eqref{eq miso ER hyper Fox's H fading} in Theorem \ref{theorem miso ER hyper Fox's H fading} can be approximated as
\begin{equation}
\label{eq approximated hyper Fox's H fading}
\begin{aligned}
R(\theta){\approx}& -\frac{1}{A} \log_2\frac{1}{\Gamma(A)}\sum_{k_1{=}1}^{K_1}\cdots\sum_{k_N{=}1}^{K_N} 
\left(\frac{Q}{M}\right)^A \\
&\times \sum_{\ell{=}1}^M  \bigg{(}{\ell}^{A{-}1} 
 e^{{-}\frac{\ell Q}{M}} \prod_{j=1}^{N} H\left[\frac{\ell Q \rho}{MN}{,}\mathbf{O}_{\phi_{j,k_j}}{,}\mathbf{P}_{\phi_{j,k_j}}\right]\bigg{)}.
\end{aligned}
\end{equation}

\section{Numerical Results}
\label{section simulation}
In order to verify the proposed MGF approach and approximation method in this paper, simulations under different fading scenarios are carried out, where the parameters are listed in Table \ref{table inid_parameters}. Without loss of generality, the duration of a time block $T=1$ ms and the bandwidth of the system $B=1$ kHz have been assumed. For each simulated scenario, we have run $10^7$ trails for simulating each channel condition with unit power. Since the parameters in hyper Fox's \textit{H} fading model have physical meanings only when specific fading models are considered, generalized $K$ fading model and Weibull/Gamma fading model have been used as examples to validate the proposed MGF method, which are all special cases of hyper Fox's \textit{H} fading model. These two fading models are very practical, which can characterize large scale and small scale fading effect simultaneously, and include many practical fading models as special cases, such as Rayleigh, Nakagami-$m$ and Weibull fading model. These two fading models have been proved to fit measurements in various channel conditions and extensivley used in the study of wireless communication systems\cite{generalized_K_fading1986,Wbl_Gamma_fading2009}. 

For a single channel, the instantaneous channel power gain parameter sequences for generalized $K$ fading model are \cite[Table IX]{H_Transform2015} $\mathbf{O}_{\gamma}=(2,0,0,2)$ and $\mathbf{P}_{\gamma}=(\frac{\hat{m}}{\psi \Gamma(\hat{m}) \Gamma(1/\psi)},\frac{\hat{m}}{\psi},{\textbf{---}},(\hat{m}-1,\frac{1}{\psi}-1),{\textbf{---}},\mathbf{1}_2)$, where multipath fading severity parameter $\hat{m} \geq \frac{1}{2}$ and shadowing figure $\psi\in [0,2]$. Specially, when $\psi=0$, the generalized $K$ fading reduces to the Nakagami-$m$ fading. 
The instantaneous channel power gain parameter sequences for Weibull/Gamma fading model are \cite[Table IX]{H_Transform2015} $\mathbf{O}_{\gamma}=(2,0,0,2)$ and $\mathbf{P}_{\gamma}=(\frac{\Gamma(1+2/\beta)}{\psi  \Gamma(1/\psi)},\frac{\Gamma(1+2/\beta)}{\psi},{\textbf{---}},(\frac{1}{\psi}-1,1-\frac{2}{\beta}),{\textbf{---}},(1,\frac{2}{\beta}))$, where the fading severity parameter $\beta > 0$ and shadowing figure $\psi \geq 0$. This model reduces to Weibull fading when $\psi=0$ and includes $K$ fading as special case when $\beta=2$. 
The exact analytical effective rate representation for i.n.i.d. conditions can be directly obtained by substituting the above parameters into the proposed MGF based approach in Theorem \ref{theorem miso ER hyper Fox's H fading}. The univariate Fox's \textit{H} function is evaluated using the method proposed in \cite{H_calculation2009}. When multivariate Fox's \textit{H} functions are involved, they are estimated using the proposed approximation method in Theorem \ref{theorem multivariate H approximation} 
while the analytical values are evaluated using the numerical method in \cite{ref_multiH_evaluation_python2016}. 

\begin{table}[!hbtp]
\centering
\caption{Parameters for different fading scenarios}
\label{table inid_parameters}
\resizebox{0.48\textwidth}{!}{%
\begin{tabular}{|l|c|c|l|}
\hline
\multicolumn{2}{|c|}{Scenarios}                 & Fading Model    & Parameters            \\ \hline
SISO                     & $N$=1                  & Weibull/Gamma   & $\beta$=3,\,\,\,\,$\psi$=1   \\ \hline
I.i.d.                    & $N$=9                  & Generalized $K$ & $\hat{m}$=2, $\psi$=$\frac{1}{1.45}$    \\ \hline
\multirow{5}{*}{i.n.i.d.} & \multirow{2}{*}{$N$=2} & Weibull/Gamma   & $\beta$=3,\,\,\,\,$\psi$=1   \\ \cline{3-4} 
                         &                      & Weibull/Gamma   & $\beta$=2,\,\,\,\,$\psi$=0.5 \\ \cline{2-4} 
                         & \multirow{3}{*}{$N$=3} & Weibull/Gamma   & $\beta$=3,\,\,\,\,$\psi$=1   \\ \cline{3-4} 
                         &                      & Weibull/Gamma   & $\beta$=2,\,\,\,\,$\psi$=0.5 \\ \cline{3-4} 
                         &                      & Generalized $K$ & $\hat{m}$=2, $\psi$=0.5       \\ \hline 
\multirow{2}{*}{Correlated} & \multirow{2}{*}{$N$=2} & \multirow{2}{*}{Correlated generalized $K$} & \multirow{2}{*}{\begin{tabular}[c]{@{}l@{}}$\hat{m}$=1, $\psi$=1\\ r=0, 0.5, 0.8 \\ \end{tabular}} \\ 
{}&{}&{}&{}\\\hline 
\end{tabular}%
}
\end{table}


\begin{figure}[!htbp]
\centering
\includegraphics[width=0.48\textwidth]{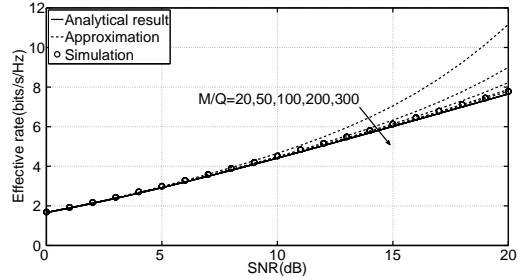}
\caption{Effective rate using different discretization and truncation ratio $\frac{M}{Q}$ compared with analytical results and simulation results in i.i.d. generalized $K$ fading scenario.}
\label{fig_M2Q_v1}
\end{figure} 

The i.i.d. conditions are special cases in i.n.i.d. conditions. In order to verify the proposed approximation methods with existing results, the analytical results are estimated using \cite[eq.(53)]{MISOFramework2012} and the parameters used in \cite{MISOFramework2012} are exploited. When approximating the effective rate using \eqref{eq trapezodial_approx_overall} in Theorem \ref{theorem trapezodial_approx ER}, the truncation parameter $Q$ can be chosen by a small number in practice, which is due to the fact that the estimated integration involving exponential fading terms. From the simulation experiments, we find that $Q=15$ is good enough for the estimation purpose, where further increase of the value of $Q$ will not give more accurate results. By increasing the ratio of the discretization parameter $M$ and the truncation parameter $Q$, the approximation approaches the exact value, as shown in Fig. \ref{fig_M2Q_v1}. It is shown that when $M/Q\geq 200$, the approximations are tight with the simulation as well as analytical results, which supports the validation of the proposed approximation method. In the following part, $M/Q$ is selected as 300 as default, which gives good accuracy as well as low computational complexity. 

\begin{figure}[!htbp]
\centering
\includegraphics[width=0.48\textwidth]{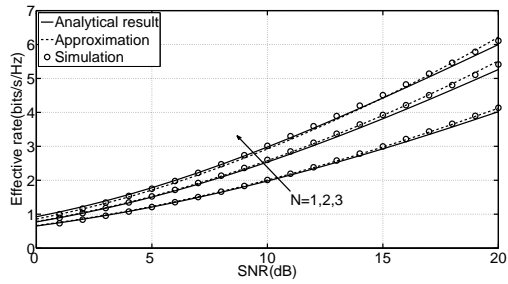}
\caption{Effective rate over SISO and i.n.i.d. MISO fading channels, where parameters for each channel are listed in Table \ref{table inid_parameters}.}
\label{fig_ER_inid}
\end{figure}

\begin{figure}[!htbp]
\centering
\includegraphics[width=0.48\textwidth]{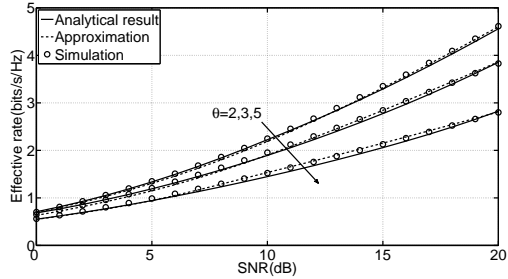}
\caption{Effective rate under different QoS exponential $\theta$ in i.n.i.d. Weibull/Gamma fading channels with $N=2$.}
\label{fig_ER_A}
\end{figure}

In order to test the performance of the proposed methods under i.n.i.d. MISO fading channel conditions, different fading parameters and channel numbers are used, which are detailed in Table \ref{table inid_parameters}. For different scenarios, it is shown in Fig. \ref{fig_ER_inid} that the approximations are sufficiently tight across a wide range of SNR (in this case from 0 to 20 dB) under different scenarios. Since the effective rate over i.n.i.d. MISO fading channels can be estimated based on the product of the individual channel's MGF as presented in Thereom \ref{theorem trapezodial_approx ER}, the proposed MGF based approach is flexible and easy to extend.

Effective rate under different QoS exponent $\theta$ is the most interested parameter in the applications, since it can be used as a metric in the QoS provisioning schemes. Larger QoS exponent $\theta$ corresponds to tighter delay constraint. As shown in Fig. \ref{fig_ER_A}, the maximum available data rate drops with the increase of QoS exponent $\theta$, in order to guarantee the system's delay performance. Also when the SNR gets higher, the same increase of $\theta$ results in greater drop of the effective rate. 



\begin{figure}[!htbp]
\centering
\includegraphics[width=0.48\textwidth]{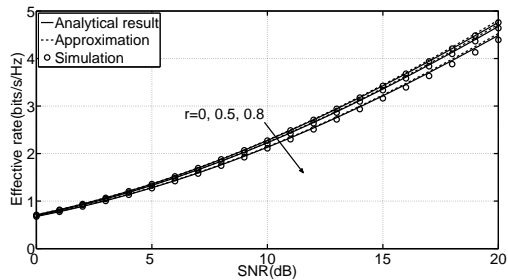}
\caption{Effective rate under correlated generalized $K$ fading channel with different correlation coefficients $r=0,0.5,0.8$ and $N=2$, where approximation results are compared to analytical results and simulation results.}
\label{fig_ER_correlated}
\end{figure}
Furthermore, the effective rate over correlated generalized $K$ fading channels is investigated, whose channel fading parameters are given in Table \ref{table inid_parameters}. It is shown in Fig. \ref{fig_ER_correlated} that the effective rate reduces as the correlation coefficient $r$ increases from 0 to 0.8, where the case of $r=0$ corresponds to i.i.d. generalized $K$ fading channels. 



\section{Conclusion}
\label{section conclusion}
In this paper, a new MGF based approach for the effective rate analysis over arbitrary correlated and not necessarily identical MISO fading channels using \textit{H} transform representation has been proposed. The proposed framework has simplified the representation and analyses of effective rate in a unified way. The effective rate over i.n.i.d. hyper Fox's \textit{H} fading channels as well as arbitrary correlated generalized $K$ fading channels has been investigated, which has given good demonstrations on the application of the proposed MGF based approaches. Based on these results, the effective rate over many practical fading channels can be obtained by simply substituting the corresponding parameters instead of the cumbersome case-by-case integration procedure.

In addition, approximations have been given for the MGF based effective rate representation as well as the effective rate representation over i.n.i.d. MISO hyper Fox's \textit{H} fading channels, where both the truncation error and discretization error have been studied. These results have been illustrated readily applicable to practical fading channels, such as Weibull/Gamma and generalized $K$ fading channels. The simulations have been used to show the validation and accuracy of the proposed analytical and approximation methods. These results have extended and complemented the existing research of  effective rate analysis.

We highlight that, as various metrics in wireless communication networks can be represented in similar \textit{H} transform format and multivariate Fox's \textit{H} functions have the possibility to characterize the statistical properties of both independent and correlated channels, the obtained results will be also valuable to the performance analyses of other statistical metrics in wireless systems.

\appendix
\subsection{Fox's \textit{H} function, multivariate Fox's \textit{H} function and \textit{H} transform}
\label{subsection Notations and Symbols}
%
The Fox's \textit{H} function \cite[Ch. 1.2]{HFunction_thebook2010} $H^{m{,}n}_{p{,}q}\left[\cdot\right]$, or univariate Fox's \textit{H} function in order to distinguish from multivariate Fox's \textit{H} function, can be defined by a single Mellin-Barnes type of contour integral as \cite[Ch. 1.2]{HFunction_thebook2010}
\begin{equation}
\label{eq H_function_notation}
H\left[s,\mathbf{O},\mathbf{P} \right]{=}uH^{m{,}n}_{p{,}q}\left[vs{\bigg|}
\begin{matrix}
\mathbf{c},\mathbf{C}\\
\mathbf{d},\mathbf{D}
\end{matrix}
\right]
{=}\frac{u}{2\pi i}\oint_\mathcal{L} \Phi \left(z\right) (vs)^z dz,
\end{equation}
where $i=\sqrt{-1}$, $s\neq 0$ and $\mathcal{L}$ is a suitable contour. The following notation is used for simplicity,
\begin{equation}
\label{eq HFunction_parameters}
\left\lbrace
\begin{aligned}
\mathbf{O}&=\left(m,n,p,q\right)\\
\mathbf{P}&=\left(u,v,\mathbf{c},\mathbf{d},\mathbf{C},\mathbf{D}\right)
\end{aligned}\right.
\end{equation}
where $\mathbf{c}{=}(\dot{\mathbf{c}}{,}\ddot{\mathbf{c}})$, $\mathbf{d}{=}(\dot{\mathbf{d}},\ddot{\mathbf{d}})$, $\mathbf{C}{=}(\dot{\mathbf{C}}{,}\ddot{\mathbf{C}})$ and $\mathbf{D}{=}(\dot{\mathbf{D}}{,}\ddot{\mathbf{D}})$, with
$\dot{\mathbf{c}}{=}(c_1{,}\dots{,}c_n)$, $\ddot{\mathbf{c}}{=}(c_{n+1}{,}\dots{,}c_p)$, 
$\dot{\mathbf{d}}{=}(d_1{,}\dots{,}d_m)$, $\ddot{\mathbf{d}}{=}(d_{m+1}{,}\dots{,}d_q)$,
$\dot{\mathbf{C}}{=}(C_1{,}\dots{,}C_n)$, $\ddot{\mathbf{C}}{=}(C_{n+1}{,}\dots{,}C_p)$,
$\dot{\mathbf{D}}{=}(D_1{,}\dots{,}D_m)$ and $\ddot{\mathbf{D}}{=}(D_{m+1}{,}\dots{,}D_q)$. Also
\begin{equation}
\label{eq theta_definition}
\Phi \left(z\right){=}\frac{\prod _{j{=}1}^m \Gamma \left( d_j{-}D_jz \right)\prod _{j{=}1}^n \Gamma \left( 1{-}c_j{+}C_jz \right)}{\prod _{j{=}n{+}1}^p \Gamma \left( c_j{-}C_jz \right)\prod _{j{=}m{+}1}^q \Gamma \left( 1{-}d_j{+}D_jz \right)}.
\end{equation}

Typical operations used in \textit{H} transform are defined as follows\cite{H_Transform2015}. The convolution operation is defined by $\mathbf{O}_1 \boxplus \mathbf{O}_2=(m_1+m_2,n_1+n_2,p_1+p_2,q_1+q_2)$ and $\mathbf{P}_1 \boxplus \mathbf{P}_2=(u_1u_2,v_1v_2,\mathbf{c}',\mathbf{d}',\mathbf{C}',\mathbf{D}')$, where $\mathbf{c}'=(\dot{\mathbf{c}}_1,\mathbf{c}_2,\ddot{\mathbf{c}}_1)$, 
$\mathbf{d}'=(\dot{\mathbf{d}}_1,\mathbf{d}_2,\ddot{\mathbf{d}}_1)$,
$\mathbf{C}'=(\dot{\mathbf{C}}_1,\mathbf{C}_2,\ddot{\mathbf{C}}_1)$ and 
$\mathbf{D}'=(\dot{\mathbf{D}}_1,\mathbf{D}_2,\ddot{\mathbf{D}}_1)$.
The Mellin operation is defined by 
$\mathbf{O}_1\boxdot\mathbf{O}_2=(m_1+n_2,m_2+n_1,p_1+q_2,p_2+q_1)$ and 
$\mathbf{P}_1 \boxdot \mathbf{P}_2=(\frac{u_1u_2}{v_2},\frac{v_1}{v_2},\mathbf{c}'',\mathbf{d}'',\mathbf{C}'',\mathbf{D}'')$, 
where 
$\mathbf{c}''=(\dot{\mathbf{c}}_1,\mathbf{1}_{q_2}-\mathbf{d}_2-\mathbf{D}_2,\ddot{\mathbf{c}}_1)$, 
$\mathbf{d}''=(\dot{\mathbf{d}}_1,\mathbf{1}_{p_2}-\mathbf{c}_2-\mathbf{C}_2,\ddot{\mathbf{d}}_1)$,
$\mathbf{C}''=(\dot{\mathbf{C}}_1,\mathbf{D}_2,\ddot{\mathbf{C}}_1)$ and
$\mathbf{D}''=(\dot{\mathbf{D}}_1,\mathbf{C}_2,\ddot{\mathbf{D}}_1)$. 
Moreover, the scaling operation $\mathbf{P}|\alpha \rangle $ is defined by $\left(u/\alpha,v/\alpha,\mathbf{c},\mathbf{d},\mathbf{C},\mathbf{D} \right)$, 
while the elementary operation $\langle \alpha,\beta,\gamma | \mathbf{P}$ is defined by $\left(u/{(\alpha v)}^{\beta\gamma},{(\alpha v)}^{\beta},\mathbf{c}+\beta\gamma\mathbf{C},\mathbf{d}+\beta\gamma\mathbf{D}, \beta \mathbf{C},\beta \mathbf{D} \right)$. Interested reader should refer to \cite{H_Transform2015} for full details.%

The operator $H^{0{,}n_0{;}m_1{,}n_1{;}\dots{;}m_N{,}n_N}  _{p_0{,}q_0{;}p_1{,}q_1{;}\dots{;}p_N{,}q_N}\left[\cdot\right]$ is Fox's \textit{H} function of $N$ variables\cite[Appendix A.1]{HFunction_thebook2010}, or multivariate Fox's \textit{H} function for simplicity, which can be defined in terms of multiple Mellin-Barnes type contour integrals as 
\begin{align}
\label{eq multivariate_H_function_definition}
&H^{0{,}n_0}_{p_0{,}q_0}\!\left[ \begin{matrix}
(a_j{;}A_j^{\left(1\right)}{,}{\dots}{,}A_j^{\left(N\right)})_{1{,}p_0}\\
(b_j{;}B_j^{\left(1\right)}{,}{\dots}{,}B_j^{\left(N\right)})_{1{,}q_0}
\end{matrix}
{:}\left(s_j{,}\mathbf{O}_j{,}\mathbf{P}_j\right)_{1{,}N}\right]\\ \notag
{=}&\frac{u_1{\cdots}u_N}{(2\pi i)^N} \int_{\mathcal{L}_{1}}\!\!\!{\cdots}\!\!\!\int_{\mathcal{L}_N} \! \Psi(\boldsymbol{\zeta})\left\lbrace \prod ^N_{j=1} \Phi_j(\zeta_j)(v_j s_j)^{\zeta_j} \right\rbrace d\zeta_1{\cdots}d\zeta_N,
\end{align}
%
 where
\begin{equation}
\label{eq Psi_definition}
\begin{aligned}
\Psi(\boldsymbol{\zeta}){=}
\frac{\prod \limits _{j{=}1}^{n_0} \!\Gamma(1{-}a_j{+}\sum \limits_{\ell{=}1}^N \! A_j^{(\ell)}\zeta_\ell)}{\prod \limits _{j{=}n_0{+}1}^{p_0}\! \Gamma(a_j{-}\!\sum \limits _{\ell{=}1}^N A_j^{(\ell)}\!\zeta_\ell)\prod \limits_{j{=}1}^{q_0} \!\Gamma(1{-}b_j{+}\!\sum\limits_{\ell{=}1}^N B_j^{(\ell)}\zeta_\ell)},
\end{aligned}
\end{equation}
\begin{equation}
\label{eq Phi_definition}
\begin{aligned}
\Phi_j(\zeta_j){=}
\frac{\prod \limits _{\ell{=}1}^{m_j} \Gamma(d_\ell^{(j)}{-}D_\ell^{(j)}\zeta_j)\prod \limits _{\ell{=}1}^{n_j} \Gamma(1{-}c_\ell^{(j)}{+}C_\ell^{(j)}\zeta_j)}{\prod \limits _{\ell{=}n_j{+}1}^{p_j} \Gamma(c_\ell^{(j)}{-}C_\ell^{(j)}\zeta_j)\prod \limits _{\ell{=}m_j{+}1}^{q_j} \Gamma(1{-}d_\ell^{(j)}{+}D_\ell^{(j)}\zeta_j)},
\end{aligned}
\end{equation}
whereas $\mathcal{L}_j$ is the suitable contours in the $\zeta_j$-plane. $(a_j;A_j^{(1)},\dots,A_j^{\left(N\right)})_{1,p_0}$ abbreviates $p_0$-parameter array $(a_1;A_1^{\left(1\right)},\dots,A_1^{\left(N\right)}),\dots,(a_{p_0};A_{p_0}^{\left(1\right)},\dots,A_{p_0}^{\left(N\right)})$, and
 $\left(s_j,\mathbf{O}_j,\mathbf{P}_j\right)_{1,N}$ abbreviates $N$-parameter array $\left\{ s_1,\mathbf{O}_1,\mathbf{P}_1;\cdots;s_N,\mathbf{O}_N,\mathbf{P}_N \right\}$, where $\mathbf{O}_j=(m_j,n_j,p_j,q_j)$ and\\
  $\mathbf{P}_j=(u_j,v_j,\mathbf{c}^{(j)},\mathbf{d}^{(j)},\mathbf{C}^{(j)},\mathbf{D}^{(j)})$,
whereas $\mathbf{c}^{(j)}$ abbreviates $p_j$-parameter array $(c_1^{(j)},\dots,c_{p_j}^{(j)})$. Other abbreviations follow the same way. See \cite{HFunction_thebook2010} for more related details.

Specially, when $ n_0 = p_0 = q_0 =0 $, the multivariate Fox's \textit{H} function breaks up into the product of $N$ univariate Fox's \textit{H} functions as \cite[eq.(A.13)]{HFunction_thebook2010}
\begin{equation}
\label{eq definition of product of H_functions}
\begin{aligned}
H^{0{,}0}_{0{,}0}\left[
\begin{matrix}
{\textbf{---}}\\
{\textbf{---}}
\end{matrix}
{:}\left(s_j{,}\mathbf{O}_j{,}\mathbf{P}_j\right)_{1{,}N}\right]
{=}\prod _{j{=}1}^N H\left[s_j{,}\mathbf{O}_j{,}\mathbf{P}_j\right].
\end{aligned}
\end{equation}

The \textit{H} transform of a function $f(z)$ with Fox's \textit{H} kernel is defined by \cite{H_Transform2015}
\begin{equation}
\label{eq H Transform definition}
\mathbb{H}\left\{f(z),\mathbf{O},\mathbf{P}\right\}(s)
{=}u\int_0^\infty \! H_{p{,}q}^{m{,}n}\left[vsz{\bigg|}\begin{matrix}
\mathbf{c},\mathbf{C}\\
\mathbf{d},\mathbf{D}
\end{matrix}\right]f(z)dz,
\end{equation}
given that the integral converges absolutely and $s>0$.
\subsection{Proof of Theorem \ref{theorem MISO_MGF_H-transform}}
\label{app proof MISO_MGF_H-transform}
Using \cite[eq.(28)]{H_Transform2015} with the identity of \cite[eq.(1.39)]{HFunction_thebook2010}, \cite[eq.(1.43)]{HFunction_thebook2010} and \cite[eq.(2.22)]{HFunction_thebook2010} we have
\begin{equation}
\label{eq MGF_identity}
\begin{aligned}
&H^{1{,}1}_{1{,}1}\left[\frac{\rho\sum_{j=1}^N\gamma_j}{N}{\bigg|}\begin{matrix}
1{-}A{,}1\\
0{,}1
\end{matrix}\right]\\
{=}&\int_0^\infty\!\! H^{1{,}0}_{0{,}1}\left[s{\bigg|}\begin{matrix}
{\textbf{---}}\\
A{-}1{,}1
\end{matrix}\right]
H^{1{,}0}_{0{,}1}\left[\frac{\rho s\sum_{j=1}^N\gamma_j}{N}{\bigg|}\begin{matrix}
{\textbf{---}}\\
0{,}1
\end{matrix}\right]ds.
\end{aligned}
\end{equation}
Substituting \cite[eq.(1.43)]{HFunction_thebook2010} into \eqref{eq ER_theory} and using \eqref{eq MGF_identity} as well as \cite[eq.(62)]{H_Transform2015}, the following equation can be obtained
\begin{equation}
\label{eq ER_MGF_int}
\begin{aligned}
R(\theta)
&{=}{-}\frac{1}{A} \log_2 \frac{1}{\Gamma(A)} \int_0^\infty H^{1{,}0}_{0{,}1}\left[s{\bigg|}\begin{matrix}
{\textbf{---}}\\
A{-}1{,}1
\end{matrix}\right] \phi_{\text{end}}\left(\frac{\rho s}{N}\right) ds.
\end{aligned}
\end{equation}
Then use the definition of \textit{H} transform \eqref{eq H Transform definition}, \eqref{eq ER MGF H transform form s} can be achieved. By changing the integral variate, an alternative form can be obtained as \eqref{eq ER MGF H transform form identical}.
%
%
%

%
\subsection{Proof of Lemma \ref{lemma T approximation}}
\label{app proof T approximation}
%
Using the \textit{H} transform definition \eqref{eq H Transform definition}, the integration form can be obtained.
First we prove that the integration can be truncated. 
Use the fact $\frac {\partial\phi(s)}{\partial s}<0$ for $s\in[0,\infty)$, the truncation error $E(Q)$ can be upper bounded by
\begin{equation}
\label{eq E_Q}
\begin{aligned}
E(Q)&=\int_Q^\infty e^{{-}s} s^{A{-}1}\phi_{{\text{end}}}(\frac{s}{N})ds\leq\Gamma(A,Q),
\end{aligned}
\end{equation}
where $\Gamma(A,Q)$ is the incomplete gamma function \cite[eq.(8.350.2)]{TheTable}. Since the integrands in \eqref{eq E_Q} are all positive for $s\in[0,\infty)$, then $E(Q)\geq 0$. 
Apply the Trapezoidal rules \cite[eq.(3.5.2)]{MathFunctions2010} to evaluate the finite integration, then \eqref{eq integration_approximation_ER} and \eqref{eq discretization_error EM} can be obtained. Follow the limitation rules, we get \eqref{eq EQ=0}.
\bibliographystyle{IEEEtran}
\bibliography{ref_EffRateMGF2015_ACCEPT_Version}

\begin{thebibliography}{10}
\providecommand{\url}[1]{#1}
\csname url@samestyle\endcsname
\providecommand{\newblock}{\relax}
\providecommand{\bibinfo}[2]{#2}
\providecommand{\BIBentrySTDinterwordspacing}{\spaceskip=0pt\relax}
\providecommand{\BIBentryALTinterwordstretchfactor}{4}
\providecommand{\BIBentryALTinterwordspacing}{\spaceskip=\fontdimen2\font plus
\BIBentryALTinterwordstretchfactor\fontdimen3\font minus
  \fontdimen4\font\relax}
\providecommand{\BIBforeignlanguage}[2]{{%
\expandafter\ifx\csname l@#1\endcsname\relax
\typeout{** WARNING: IEEEtran.bst: No hyphenation pattern has been}%
\typeout{** loaded for the language `#1'. Using the pattern for}%
\typeout{** the default language instead.}%
\else
\language=\csname l@#1\endcsname
\fi
#2}}
\providecommand{\BIBdecl}{\relax}
\BIBdecl

\bibitem{EffCap}
D.~Wu and R.~Negi, ``Effective capacity: {A} wireless link model for support of
  quality of service,'' \emph{IEEE Trans. Wireless Commun.}, vol.~2, no.~4, pp.
  630--643, Jul. 2003.

\bibitem{E2E_EffCap2012}
A.~Khalek and Z.~Dawy, ``Energy-efficient cooperative video distribution with
  statistical {QoS} provisions over wireless networks,'' \emph{IEEE Trans.
  Mobile Comput.}, vol.~11, no.~7, pp. 1223--1236, Jul. 2012.

\bibitem{EC_TD_Scheduling2010}
A.~Balasubramanian and S.~Miller, ``The effective capacity of a time division
  downlink scheduling system,'' \emph{IEEE Trans. Commun.}, vol.~58, no.~1, pp.
  73--78, Jan. 2010.

\bibitem{ref_ER_TD2015}
S.~Agarwal, S.~De, S.~Kumar, and H.~Gupta, ``{QoS}-aware downlink cooperation
  for cell-edge and handoff users,'' \emph{IEEE Trans. Veh. Technol.}, vol.~64,
  no.~6, pp. 2512--2527, Jun. 2015.

\bibitem{EC_2Hop2013}
D.~Qiao, M.~Gursoy, and S.~Velipasalar, ``Effective capacity of two-hop
  wireless communication systems,'' \emph{IEEE Trans. Inf. Theory}, vol.~59,
  no.~2, pp. 873--885, Feb. 2013.

\bibitem{EffCap_CR_2010}
S.~Akin and M.~C. Gursoy, ``Effective capacity analysis of cognitive radio
  channels for quality of service provisioning,'' \emph{IEEE Trans. Commun.},
  vol.~9, no.~11, pp. 3354--3364, Nov. 2010.

\bibitem{ref_EC_CRN2015}
Y.~Yang, S.~Aissa, and K.~Salama, ``Spectrum band selection in delay-{QoS}
  constrained cognitive radio networks,'' \emph{IEEE Trans. Veh. Technol.},
  vol.~64, no.~7, pp. 2925--2937, Jul. 2015.

\bibitem{MISOFramework2012}
M.~Matthaiou, G.~C. Alexandropoulos, H.~Q. Ngo, and E.~G. Larsson, ``Analytic
  framework for the effective rate of {MISO} fading channels,'' \emph{IEEE
  Trans. Commun.}, vol.~60, no.~6, pp. 1741--1751, Jun. 2012.

\bibitem{ER_Weibull2015}
M.~You, H.~Sun, J.~Jiang, and J.~Zhang, ``Effective rate analysis in {Weibull}
  fading channels,'' \emph{IEEE Wireless Commun. Lett.}, vol.~5, no.~4, pp.
  340--343, Apr. 2016.

\bibitem{eta_mu_EffRate_2013}
J.~Zhang, M.~Matthaiou, Z.~Tan, and H.~Wang, ``Effective rate analysis of
  {MISO} $\eta$-$\mu$ fading channels,'' in \emph{Proc. IEEE Int. Conf. Commun.
  (ICC)}, Budapest, Jun. 2013, pp. 5840--5844.

\bibitem{alpha_mu_EffRate_2015}
J.~Zhang, L.~Dai, Z.~Wang, D.~Ng, and W.~Gerstacker, ``Effective rate analysis
  of {MISO} systems over $\alpha$-$\mu$ fading channels,'' in \emph{Proc.
  Global Communications Conference (GLOBECOM)}, San Diego, CA, Dec. 2015, pp.
  1--6.

\bibitem{kappa_mu_EffRate_2015}
J.~Zhang, L.~Dai, W.~H. Gerstacker, and Z.~Wang, ``Effective capacity of
  communication systems over $\kappa{-}\mu$ shadowed fading channels,''
  \emph{Electron. Lett.}, vol.~51, no.~19, pp. 1540--1942, Sep. 2015.

\bibitem{correlated_MISO_EffRate}
C.~Zhong, T.~Ratnarajah, K.~K. Wong, and M.~S. Alouini, ``Effective capacity of
  correlated {MISO} channels,'' in \emph{Proc. IEEE Int. Conf. Commun. (ICC)},
  Kyoto, Jun. 2011, pp. 1--5.

\bibitem{ref_correlated_Nakagami2012}
X.~B. Guo, L.~Dong, and H.~Yang, ``Performance analysis for effective rate of
  correlated {MISO} fading channels,'' \emph{Electron. Lett.}, vol.~48, no.~24,
  pp. 1564--1565, November 2012.

\bibitem{MGF_EGC_MRC_GFading_ShannonC_2012}
F.~Yilmaz and M.~Alouini, ``A unified {MGF}-based capacity analysis of
  diversity combiners over generalized fading channels,'' \emph{IEEE Trans.
  Commun.}, vol.~60, no.~3, pp. 862--875, Mar. 2012.

\bibitem{ShannonC_MGF_2010}
M.~Di~Renzo, F.~Graziosi, and F.~Santucci, ``Channel capacity over generalized
  fading channels: A novel {MGF}-based approach for performance analysis and
  design of wireless communication systems,'' \emph{IEEE Trans. Veh. Technol.},
  vol.~59, no.~1, pp. 127--149, Jan. 2010.

\bibitem{ref_MGF_Pout2000}
Y.~C.~K. Ko, M.~S. Alouini, and M.~K. Simon, ``Outage probability of diversity
  systems over generalized fading channels,'' \emph{IEEE Trans. Commun.},
  vol.~48, no.~11, pp. 1783--1787, Nov. 2000.

\bibitem{ER_SISO_MGF_H_2014}
Z.~Ji, C.~Dong, Y.~Wang, and J.~Lu, ``On the analysis of effective capacity
  over generalized fading channels,'' in \emph{Proc. IEEE Int. Conf. Commu.
  (ICC)}, Sydney, Jun. 2014, pp. 1977--1983.

\bibitem{ref_MGF_ER_MultiHop_2016}
K.~Peppas, P.~T. Mathiopoulos, and J.~Yang, ``On the effective capacity of
  amplify-and-forward multi-hop transmission over arbitrary and correlated
  fading channels,'' \emph{IEEE Wireless Commun. Lett.}, vol.~PP, no.~99, pp.
  1--1, 2016.

\bibitem{H_Transform2015}
Y.~Jeong, H.~Shin, and M.~Win, ``{H}-transforms for wireless communication,''
  \emph{IEEE Trans. Inf. Theory}, vol.~61, no.~7, pp. 3773--3809, Jul. 2015.

\bibitem{MathFunctions2010}
F.~W. Olver, D.~W. Lozier, R.~F. Boisvert, and C.~W. Clark, \emph{{NIST}
  handbook of mathematical functions}.\hskip 1em plus 0.5em minus 0.4em\relax
  Washington, DC: Cambridge University Press, 2010.

\bibitem{lowSNR2011}
M.~C. Gursoy, ``{MIMO} wireless communications under statistical queueing
  constraints,'' \emph{IEEE Trans. Inf. Theory}, vol.~57, no.~9, pp.
  5897--5917, Sep. 2011.

\bibitem{HFunction_thebook2010}
A.~M. Mathai, R.~K. Saxena, and H.~J. Haubold, \emph{The {H}-function: theory
  and applications}, 2010th~ed.\hskip 1em plus 0.5em minus 0.4em\relax Springer
  Science \& Business Media, Oct. 2009.

\bibitem{HyperFoxHFading2012}
F.~Yilmaz and M.~S. Alouini, ``A novel unified expression for the capacity and
  bit error probability of wireless communication systems over generalized
  fading channels,'' \emph{IEEE Trans. Commun.}, vol.~60, no.~7, pp.
  1862--1876, Jul. 2012.

\bibitem{kappa_mu_inid_2014}
J.~Zhang, Z.~Tan, H.~Wang, Q.~Huang, and L.~Hanzo, ``The effective throughput
  of {MISO} systems over $\kappa-\mu$ fading channels,'' \emph{IEEE Trans. Veh.
  Technol.}, vol.~63, no.~2, pp. 943--947, Feb. 2014.

\bibitem{Cook_H_thesis1992}
J.~I.~D. Cook, ``The {H}-function and probability density functions of certain
  algebraic combinations of independent random variables with {H}-function
  probability distribution,'' Ph.D. dissertation, Univ. Texas, Austin, TX, USA,
  1981.

\bibitem{HVariates1972}
B.~D. Carter, ``On the probability distribution of rational functions of
  independent {H}-function variates,'' Ph.D. dissertation, Univ. Arkansas,
  Fayetteville, AR, USA, 1972.

\bibitem{forLemma11986}
H.~Srivastava and M.~Garg, ``Some integrals involving a general class of
  polynomials and the multivariable {H}-function,'' \emph{Rev. Roumaine Phys.},
  vol.~22, pp. 685--692, 1987.

\bibitem{ref_multiH_property_1977}
R.~Saxena, ``On the {H}-function of n variables,'' \emph{Kyungpook Math J},
  vol.~17, pp. 221--226, 1977.

\bibitem{ref_Hfunction_expansion2013}
S.~Gaboury and R.~Tremblay, ``An expansion theorem involving {H}-function of
  several complex variables,'' \emph{Int. J. Analysis}, 2013.

\bibitem{ref_expansion_multiH1976}
H.~Srivastava and R.~Panda, ``Expansion theorems for the {H} function of
  several complex variables,'' \emph{J. Reine Angew. Math}, vol. 288, pp.
  129--145, 1976.

\bibitem{ref_integration_multiH1983}
H.~Srivastava and N.~Singh, ``The integration of certain products of the
  multivariable {H}-function with a general class of polynomials,''
  \emph{Rendiconti del Circolo Matematico di Palermo}, vol.~32, no.~2, pp.
  157--187, 1983.

\bibitem{ref_HFading1977}
B.~D. Carter and M.~D. Springer, ``The distribution of products, quotients and
  powers of independent {H}-function variates,'' \emph{SIAM J. Appl. Math.},
  vol.~33, no.~4, pp. 542--558, Aug. 1977.

\bibitem{VTC_recommend_Weibull1988}
{IEEE Vehicular Technology Society Committee on Radio Propagation}, ``Coverage
  prediction for mobile radio systems operating in the 800/900 {MHz} frequency
  range,'' \emph{IEEE Trans. Veh. Technol.}, vol.~37, no.~1, pp. 3--72, Feb.
  1988.

\bibitem{ref_reduction_hypergeometric1985}
H.~Srivastava, ``A class of generalised multiple hypergeometric series arising
  in physical and quantum chemical applications,'' \emph{J. Physics A:
  Mathematical and General}, vol.~18, no.~5, p. 227, 1985.

\bibitem{generalized_K_fading1986}
R.~Barakat, ``Weak-scatterer generalization of the {K}-density function with
  application to laser scattering in atmospheric turbulence,'' \emph{J. OSA.
  A.}, vol.~3, no.~4, pp. 401--409, Apr. 1986.

\bibitem{ref_correlated_GammaGamma}
J.~Zhang, M.~Matthaiou, G.~K. Karagiannidis, and L.~Dai, ``On the multivariate
  gamma-gamma distribution with arbitrary correlation and applications in
  wireless communications,'' \emph{IEEE Trans. Veh. Technol.}, vol.~65, no.~5,
  pp. 3834--3840, May 2016.

\bibitem{Wbl_Gamma_fading2009}
P.~Bithas, ``Weibull-gamma composite distribution: alternative
  multipath/shadowing fading model,'' \emph{Electron. Lett.}, vol.~45, no.~14,
  pp. 749--751, Jul. 2009.

\bibitem{H_calculation2009}
F.~Yilmaz and M.~S. Alouini, ``Product of the powers of generalized
  {Nakagami-$m$} variates and performance of cascaded fading channels,'' in
  \emph{Proc. IEEE Global Commun. Conf. (GLOBECOM)}, Honolulu, Nov. 2009, pp.
  1--8.

\bibitem{ref_multiH_evaluation_python2016}
H.~R. Alhennawi, M.~M. H.~E. Ayadi, M.~H. Ismail, and H.~A.~M. Mourad,
  ``Closed-form exact and asymptotic expressions for the symbol error rate and
  capacity of the {H} -function fading channel,'' \emph{IEEE Trans. Veh.
  Technol.}, vol.~65, no.~4, pp. 1957--1974, Apr. 2016.

\bibitem{TheTable}
I.~S. Gradshteyn and I.~M. Ryzhik, \emph{Table of integrals, series, and
  products}, 6th~ed.\hskip 1em plus 0.5em minus 0.4em\relax Academic Press,
  2000.

\end{thebibliography}

\end{document}